\documentclass{IEEEtran}
\usepackage{balance}
\usepackage{microtype}
\usepackage{amsmath}
\usepackage[mediummath]{nccmath}
\usepackage{amssymb}
\usepackage{stmaryrd}
\usepackage{tabularx}
\usepackage{stix}
\usepackage[caption=false, font=footnotesize]{subfig}

\usepackage{mathtools}

\usepackage{siunitx}
\usepackage{enumitem}
\usepackage[many]{tcolorbox}
\tcbuselibrary{skins}
\usepackage{lipsum}
\usepackage{graphics}
\usepackage[switch, mathlines]{lineno}
\usepackage{listings}
\usepackage[linguistics]{forest}
\usepackage{xcolor}

\usepackage[style=ieee]{biblatex}
\usepackage{ifthen}

\makeatletter
\newcounter{IEEE@bibentries}
\renewcommand\IEEEtriggeratref[1]{%
  \renewbibmacro{finentry}{%
    \stepcounter{IEEE@bibentries}%
    \ifthenelse{\equal{\value{IEEE@bibentries}}{#1}}
    {\finentry\@IEEEtriggercmd}
    {\finentry}%
  }%
}
\makeatother

\usepackage{hyperref}
\usepackage{cleveref}
\usepackage{floatflt}
\usepackage{pgfplots}
\usepackage{pgfplotstable}
\pgfplotsset{compat=1.7}
\usepackage{tikz}

\usetikzlibrary{positioning,fit,arrows.meta,arrows,backgrounds}
\usetikzlibrary{decorations.pathreplacing,calligraphy}
\usetikzlibrary{patterns}
\usepgfplotslibrary{fillbetween}

\graphicspath{{./figures/}}%
\usepackage{amsthm}

\newtheorem{definition}{Definition}
\newtheorem{theorem}{Theorem}
\newtheorem{corollary}{Corollary}

\newtheorem{remark}{Remark}
\newtheorem{proposition}{Proposition}

\newcommand{\gtuple}{\langle V, E_T, E_B,$ $\iota_T, \iota_B, m, Obs, F \rangle}

\title{On the complexity of sabotage games for~network~security} %
\author{Dhananjay Raju, Georgios Bakirtzis, Ufuk Topcu\thanks{The University of Texas at Austin, \{draju, bakirtzis, utopcu\}@utexas.edu}}

\begin{document}
\maketitle

\begin{abstract}
Securing dynamic networks against adversarial actions is challenging because of the need to anticipate and counter strategic disruptions by adversarial entities within complex network structures. Traditional game-theoretic models, while insightful, often fail to model the unpredictability and constraints of real-world threat assessment scenarios. We refine \emph{sabotage games} to reflect the realistic limitations of the saboteur and the network operator. By transforming sabotage games into reachability problems, our approach allows applying existing computational solutions to model realistic restrictions on attackers and defenders within the game. Modifying sabotage games into dynamic network security problems successfully captures the nuanced interplay of strategy and uncertainty in dynamic network security. Theoretically, we extend sabotage games to model network security contexts and thoroughly explore if the additional restrictions raise their computational complexity, often the bottleneck of game theory in practical contexts. Practically, this research sets the stage for actionable insights for developing robust defense mechanisms by understanding what risks to mitigate in dynamically changing networks under threat.
\end{abstract}

\section{Introduction}

Formal methods help develop secure systems, providing systematic guarantees for security properties, such as availability, in complex networks~\cite{schaffer:2016,terBeek:2022}. In particular, two-player graph games, effectively capturing attack patterns and defense dynamics, constitute a promising avenue for visualizing and analyzing attacker-defender scenarios~\cite{aslanyan:2016,baier:2008,kordy:2014}. In this work, we consider attacker-defender interactions in game-theoretic formulations to model network scenarios by encoding both controllable---actions by the defender---and uncontrollable---exploits by the attacker---plays in the game without making specific assumptions regarding the types of attacks~\cite{charterjee:2012}.

 For example, we can model active deception using a two-player graph game to model incomplete information. An administrator (defender) employs honey-pots, access control, and online network reconfiguration to conduct deceptive planning against malicious intrusions~\cite{horak:2017}. This model encodes a defender, who modifies a network graph to prevent an attacker from intruding into the system. We focus on security problems involving a violation of availability~\cite{bau:2011}, in which the system must guarantee the satisfaction of the desired property (B\"uchi) under node deletion by an attacker. We choose to generalize \emph{sabotage games}---a two-player graph reachability game on evolving structures~\cite{vanBenthem2005}---to address availability comprehensively in the game-theoretic context.
 
In the classical formulation of a sabotage game, one player called the \emph{traveler} representing the system tries to get from an initial vertex to a fixed set of final vertices. In contrast, the other player, the \emph{saboteur} representing the attacker, tries to prevent traveler actions by deleting edges from the graph. Sabotage games serve as an analysis framework for dynamic network assessment~\cite{catta2023attack}, for example, to design reliable networks for transportation or trustworthy information exchange~\cite{Aucher2018ModalLO, Kvasov2016Mar}.

Despite the applications of sabotage games, the classical formulation makes unrealistic assumptions that do not hold in security practice. Consider, for example, a physical security problem of determining if a possible infiltration route exists through a region demarcated by the terrain's existing roads, paths, and physical characteristics, modeled as a graph. A saboteur aims to prevent infiltration by using traps corresponding to deleting edges. If the saboteur has unrestricted mobility and no budget constraints, which means she can remove any number of edges, then the classic formulation adequately models this scenario. However, the saboteur does not have unrestricted mobility and an unlimited budget in practice. The saboteur may have to use vital resources to delete a link, thus limiting the number of edges she can delete. The classic formulation is also not suitable for modeling travelers with \emph{imperfect information}. Suppose, then, that the traveler is a robotic agent relying on its sensors; we must account for observational restrictions due to limited sensing capabilities.

We extend the classic formulation of sabotage games, address these shortcomings for dynamic network security problems, and study their potential efficacy using tools from computational complexity. In the proposed formulation of sabotage games, the saboteur and traveler start at designated vertices and move along their edges. The edge sets for the saboteur may differ from that of the traveler, allowing the saboteur to have different (maybe greater) mobility than the traveler. Because the saboteur and the traveler move on other edges, the ability of the saboteur to make topological changes to the underlying graph is accomplished by letting the saboteur delete vertices instead of edges. 

The scenarios encountered frequently resemble complex, adversarial games in network security. These situations demand not only effective solutions but also computationally feasible ones. The balance between strategic effectiveness and computational practicality is critical. In this context, our study delves into the complexities inherent in such security challenges. We focus on a classic variant of sabotage games, specifically those involving the strategic deletion of network connections.
Our approach involves rigorously modeling this modified variant of sabotage games. This modeling is crucial as it enables us to accurately capture the essence of network security scenarios where adversaries strategically disrupt connectivity. More importantly, through this modeling, we demonstrate that solving these games requires considerable complexity. We establish that solving this particular game variant falls into the category of PSPACE-hard problems.
This classification into PSPACE-hard is significant. It implies that the problem, like many real-world security challenges, may have yet to find and verify solutions, mirroring the often taxing computational resources required in network defense. Thus, our study extends theoretical understanding and sheds light on the practical aspects of implementing security strategies in complex network environments, where computational considerations are as integral as the strategies.

\vspace{1em}
\noindent
\textbf{Contributions}

We use a reachability game on a graph of size exponential in the size of the graph and budget of the saboteur to find a winning strategy for the traveler if one exists. This approach allows a uniform way to handle all the parameters in a security scenario (\cref{sec:scenarios}). Complexity results vary depending on the restrictions required for a specific security game formulation (\cref{tab:complexity}).
\begin{enumerate}[leftmargin=*]
    \item Solving the variant of sabotage games for network security is EXPTIME-complete. We show the matching lower bound for hardness by encoding the acceptance of an input string by an alternating Turing machine~(ATM)%
    using polynomial space into the existence of a winning strategy for the traveler in a sabotage game. The saboteur and traveler use their moves to mimic the moves of the $\forall$ and $\exists$ player, respectively. The ATM accepts the word only if the traveler has a strategy to reach a final vertex.
    \item When defenders have an unlimited budget, the game ends in polynomial (in the size of graph) rounds, which implies containment in PSPACE. By exploiting this bound, we provide a quantified Boolean formula (QBF) encoding for sabotage games. The PSPACE-hardness of the classical formulation of sabotage games ensures that the problem is PSPACE-complete.
\end{enumerate}

\begin{table}[t!]
\renewcommand{\arraystretch}{1.3}
    \caption{The computational complexity of finding a strategy for the traveler against a single saboteur.}
    \begin{tabularx}{\columnwidth}{ccc}
\textbf{Saboteur restriction} & \textbf{Traveler restriction} & \textbf{Computational complexity} \\[1em]
(B1)  & (T1) (T2) (T3) & PSPACE-complete  \\
(B2)  & (T1) & PTIME \\
(B2)  & (T2) & PSPACE-complete \\
(B2)  & (T3) & EXPTIME-complete \\
    \end{tabularx}
    \label{tab:complexity}
\end{table}

From the saboteur's perspective, travelers with limited sensing do not improve complexity regarding solution time. Indeed, the saboteur must proactively assume that the attacker has perfect information and may have access to an oracle that correctly reveals or guesses the deleted vertices correctly. When the saboteur has a fixed budget, the problem is solvable in PTIME (it is EXPTIME otherwise). The result of the reachability game also lets us determine whether an initial configuration of deleted vertices would allow the saboteur to win the game.

\subsection{Networks under threat} \label{sec:scenarios}
We present a few simplified scenarios to illustrate the applicability of sabotage games.

\begin{figure}
    \centering
    \includegraphics[width=.8\linewidth]{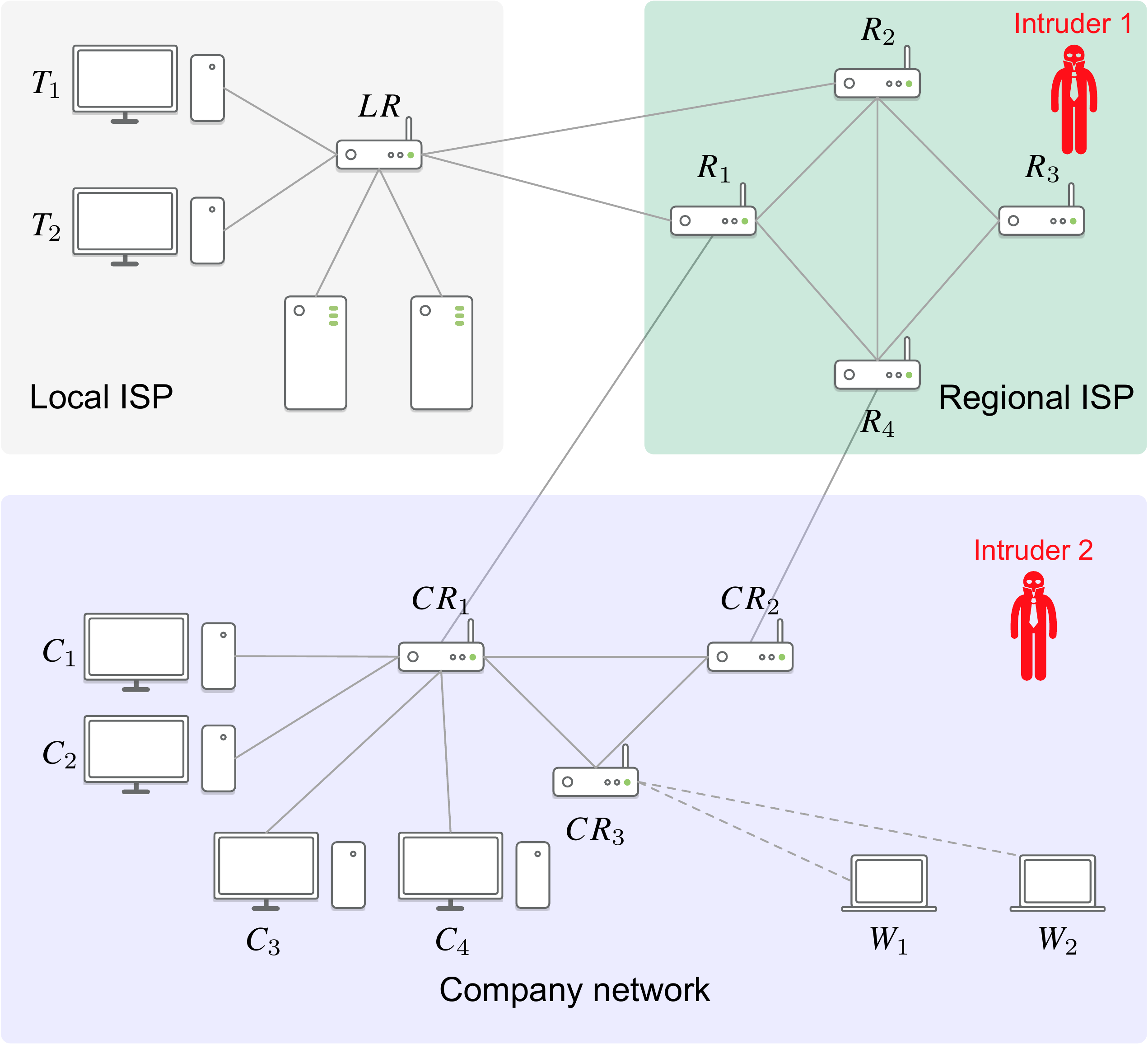}
    \caption{Intruder~1 tries to deny access to the company network by affecting the nodes in the regional internet service provider (ISP). Intruder~2 affects the company's nodes (network architecture adapted from Kurose and Ross~\cite{kurose:2001}).}
\label{fig:scenario}
\end{figure}

\paragraph{Scenario~1 [Dynamic network routing under threat]} 
We motivate the use of sabotage games using a dynamic network security scenario where intruders (corresponding to saboteurs) try denying access to company personnel (travelers). The traveler attempts to send a message to one of the terminal vertices, namely $C_1$, $C_2$, $C_3$, and $C_4$ in the company network~(\cref{fig:scenario}). Two terminals, $T_1$ and $T_2$, are connected to a local area service provider's router $LR$. This router $LR$ is connected to a regional service provider's router $R_1$ and $R_2$. The company has three routers: $CR_1, CR_2$, and $CR_3$. Two intruders operate in the regional internet service provider (ISP) and the company's network. They can disrupt the traffic by affecting the routers in the regional ISP and the company's network, respectively. Additionally, the intruder in the company network can only begin by attacking the router $CR_3$.
We analyze the availability problem under the following restrictions on the intruder (saboteur) and company user (traveler):
\begin{enumerate}
    \item[\textbf{(B1)}] There is no restriction on the number of nodes that the saboteur can affect.
    \item[\textbf{(B2)}] The number of nodes that the saboteur can affect is a fixed constant $m \in \mathbb{N}$.
    \item[\textbf{(T1)}] The traveler can observe the deleted nodes.
    \item[\textbf{(T2)}] The traveler cannot observe the deleted nodes.
    \item[\textbf{(T3)}] The traveler can observe the deleted nodes adjacent to them from an observation function.
\end{enumerate}
The administrator can use a fixed number of defense mechanisms, such as firewalls, to sabotage attacks. Additionally, for analyzing the worst case, we assume that the attacker observes the locations of the defenses.
We assume the attacker cannot observe the defenses set by the system's administrator when she is in a different location.
The fact that the traveler can observe the position of the saboteur does not automatically mean that she knows if the node has been disabled.

\paragraph{Scenario~2 [Network security]} 
We describe a scenario in which it is essential to analyze the sabotage game from the perspective of the saboteur.
A defender (saboteur) must disable a fixed number of nodes in a network represented by a directed graph to protect the network from infiltration by an intruder (traveler). The defender disables the nodes to respond to the intruder's actions (traveler) and does not disable nodes apriori. We capture budget-based constraints for the defender by requiring that it may disable at most $m$ nodes. The intruder starts from a node and either attempts to reach a final node once or visits it infinitely often. We can think of the last vertices as having some secret the defender must protect from intruders.

Consider a scenario where a  supplier robot must deliver resources to a drop-off zone in a zone represented by a directed graph. The supplier robot must avoid adversarial action while avoiding traps placed by the adversary. The adversary's objective is to prevent the supplier from reaching one of the drop zones. Like the previous scenario, the supplier (traveler) has limited observability. Unlike the last scenario network, the saboteurs do not have unrestricted mobility: \textbf{(B3)} The edge set of the saboteurs is not a complete graph.

From the traveler's perspective, the observability of defenses/security controls---deleted vertices in the game graph---changes the problem's difficulty. The attacker can actively determine whether a vertex has been deleted before traversing a path. For this reason, we also address the effect of imposing limited sensing on the traveler. 

\section{Sabotage games}
\label{sec:sabotage}

Sabotage games are played between a traveler and a saboteur on a graph with two sets of edges. The traveler moves along the traveler's edges, and the saboteur moves on the saboteur's edges. The saboteur changes the graph's structure by marking and deleting vertices. The traveler must avoid the saboteur and marked vertices while satisfying a specific objective.

\begin{definition}
\label{defn:sabotage}
A sabotage game $G$~is a tuple $\langle V, E_T, E_B,$ $\iota_T, \iota_B, m, Obs, F \rangle$ composed of the following data. 
\begin{itemize}[leftmargin=*]
    \item A set of vertices (or positions), $V$.
    \item Sets of directed edges over $V$, $E_T$ and $E_B$.
    \item Initial positions  $\iota_T$ and $\iota_B$  of the traveler and intruder (resp.) .
    \item The budget for the intruder, $m \in \mathbb{N} \cup \{ \infty \}$.
    \item The observation function $Obs\colon V \to 2^V$ for the traveler.
    \item The set of final vertices, $F \subseteq V$.
\end{itemize}
\end{definition}
\noindent Formally, a \emph{traveler} and a \emph{saboteur} play a game on a \emph{directed} graph $G = \left(V,E_T,E_B\right)$ with two edge sets from initial locations $\iota_T$ and $\iota_B$, respectively. The game is played over a sequence of discrete-time steps or \emph{rounds}, with the traveler going first. The traveler and saboteur play in alternation. The traveler can move to an adjacent vertex with respect to $E_T$ during the traveler's turn. During the saboteur's turn, she moves to an adjacent vertex with respect to $E_B$. She can mark (delete) any of the two vertices corresponding to the edge provided the vertex in question is not occupied by the traveler and has the budget to do so.
The saboteur can mark at most $m \in \mathbb{N}$ vertices, and the saboteur has \emph{unlimited budget} when $m = \infty$. The traveler \emph{loses} the game when it occupies a marked vertex. 
The objective of the traveler is to either \begin{itemize}[leftmargin=*]
    \item reach a final vertex $v_f \in F$ (\emph{reachability}) or
    \item visit a final vertex infinitely often (B\"uchi). %
\end{itemize}

A \emph{configuration} is a triple $(t,b,T)$, where $t,b \in V$ and $T \in 2^V$.
A configuration encodes the positions of the traveler, the saboteur, and the positions of the vertices marked by the saboteur. 

\begin{definition}[Play]
A \emph{play} $\rho$ in the sabotage game is a (possibly infinite) sequence $C_1C_2C_3\cdots$ of configurations, where  $C_i=\left(t^i,b^i,T^i\right)$, such that for all $i$,
the following properties hold.
\begin{enumerate}[leftmargin=*]
    \item The initial configuration is $\left(t^1,b^1,T^1\right) = \left(\iota_T,\iota_B,\emptyset\right)$.
    \item The traveler and saboteur can only traverse through their edges $\left(t^i,t^{i+1}\right) \in E_T$ and $\left(b^i,b^{i+1}\right) \in E_B$.
    \item  The saboteur can only mark the vertices $T^i \subseteq T^{i+1} \subseteq T^i \cup \left\{b^i, b^{i+1}\right\}$.
    \item  the saboteur cannot exceed the allocated budget $|T^{i}| \leq m$; .
\end{enumerate}
\end{definition}

We introduce the observation function $Obs \colon V \to 2^V$ to capture scenarios in which the traveler cannot make perfect observations. However, there is an observational disparity between the saboteur and the traveler because the saboteur has perfect information about the marked vertices. %

The saboteur can play using the complete history of the game so far; consequently, the strategy function for the defender has the following structure. Let $\mathcal{V} = V^{2} \times \left(2^V\right)$.

\begin{definition}[Strategy for the saboteur]
 A strategy $\sigma$ for the saboteur is a partial function $\sigma\colon \mathcal{V}^+ \times V \to V \times \left(2^V\right)$. 
\end{definition}
The traveler can only use a history of what she observes to decide their next move, and their strategy is as follows.

\begin{definition}[Strategy for the traveler]
 A strategy $\pi$ for the attacker is a partial function $\pi\colon \mathcal{V}^+ \to V$.
\end{definition}

A play $\rho = C_1C_2\cdots$ is said to \emph{agree} with the strategy $\sigma$ for the saboteur if for all~$i\colon\left(b^{i+1},T^{i+1}\right) = \sigma\left(C_1C_2\cdots C_i, t^{i+1}\right).$ 
Similarly, a play $\rho = C_1C_2\cdots$ is said to \emph{agree} with a strategy $\pi$ for the traveler if for all $i \colon t^{i+1} = \pi\left(obs(C_1)\cdots obs(C_{i})\right).$

In a game with a reachability objective, a play $\rho$ is \emph{winning} for the traveler if, for some $i\in \mathbb{N}$, $t^i \in F$, that is, the traveler has reached a final state. %
A strategy is \emph{winning} if all plays that agree with it are winning.

\begin{remark}
We let the \emph{traveler make the first move}. However, the order of play can be flipped by adding a dummy vertex before the start vertex for the saboteur and allowing the saboteur to go first from the new vertex.
\end{remark}

\subsection{Dynamic network routing under threat as a sabotage game} 

\begin{figure}
    \centering
    \centering
    \includegraphics[width=.8\linewidth]{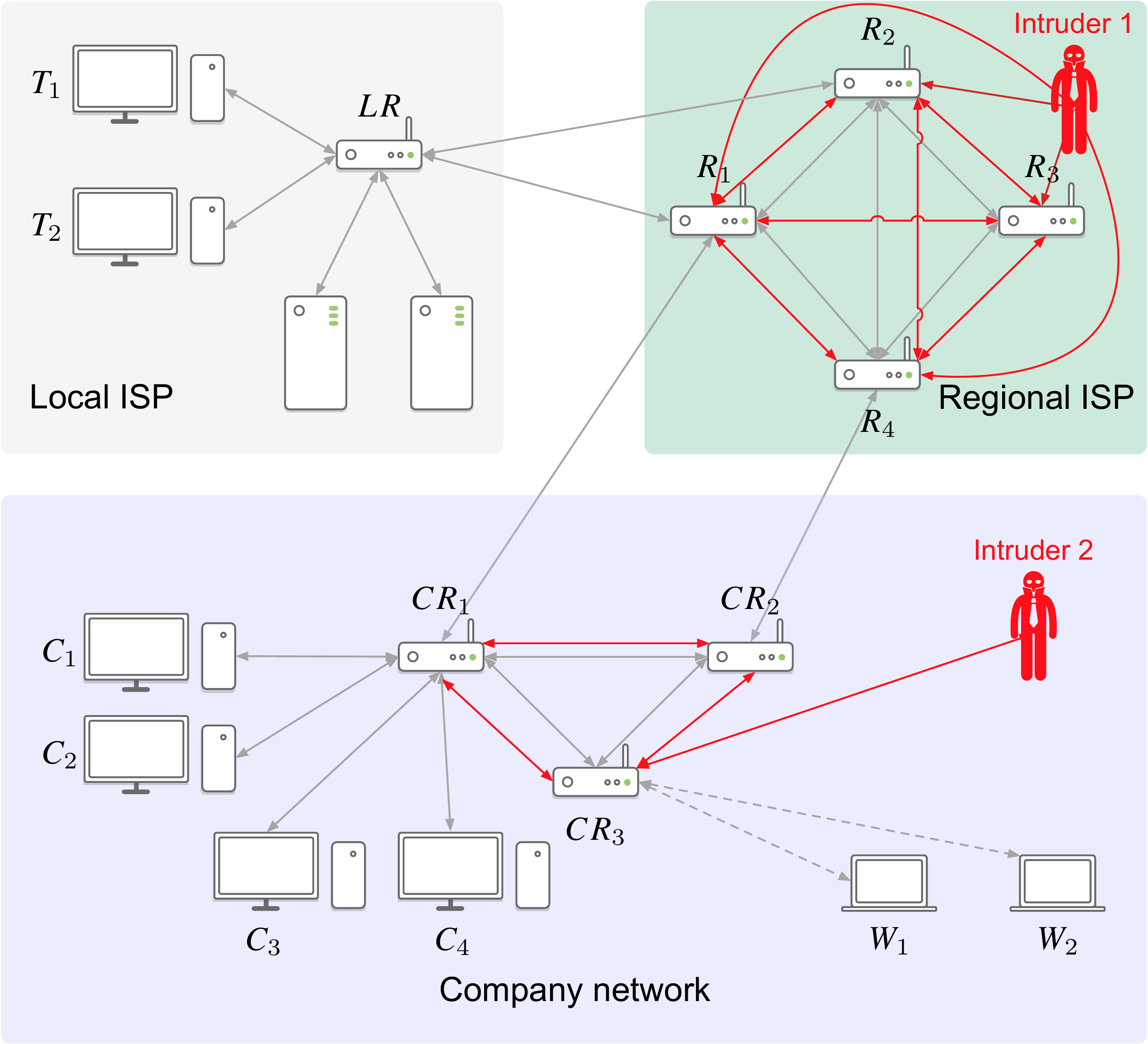}
    \caption{The sabotage game for the attack scenario in \cref{fig:scenario}.}
    \label{fig:sabotage_game_scenario}
\end{figure}

The security scenario (\cref{fig:scenario}) can be modeled as a sabotage game (\cref{fig:sabotage_game_scenario}) on the graph structure $G = \gtuple$ containing the following data.
\begin{itemize}[leftmargin=*]
    \item The vertices $V$ of the sabotage game correspond to the various components in the network. Two copies are introduced for each server, corresponding to the attacker using a mail or USB to infect the server.
    We introduce two dummy vertices to serve as the starting points for the two saboteurs, that is, the intruders.
    \item  The actions $E_T$ of the traveler are represented using black lines, and the actions $E_B$ of the saboteur are represented using red lines.
    \item The saboteurs are initially marked vertically, and the traveler's message is at the router at the local ISP ($LR$).
\end{itemize}
The saboteurs can prevent message transmission by disabling the routers via an exploit. Often, we cannot determine the intruders' following attack location. Thus, the edge set $E_B$ of the saboteur corresponds to a total graph that captures the lack of attack location knowledge. In other scenarios, where more information about the intruder is known (for example, attack patterns), we can appropriately remove the corresponding red-colored edges. We analyze the scenario where the message sent by the personnel is already at the local router $LR$. 
Without loss of generality, we assume that the saboteur cannot delete a vertex currently occupied by the traveler.

We analyze this scenario first assuming that only one intruder is operating the regional ISP and secondly with both intruders at the regional ISP network and company network (\cref{fig:sabotage_game_scenario}).

\paragraph{Intruder 1.}
When there is no limit on the number of vertices, the intruder can disable (\textbf{B1}), and the traveler has no visibility regarding these deleted vertices, there is no way for the message to be routed to one of the terminals in the company network. 

Consider the same scenario when the intruder (saboteur) has a budget of just one deletion, then we see that observability of the deletion plays a role. 
If the deletion is not observable, the traveler must proactively avoid suspect nodes (\textbf{T2}). In this scenario, the traveler does not have a winning strategy. However, when the traveler has complete visibility or can observe deleted vertices from adjacent vertices, the traveler has a way to route the message.

\begin{figure}
    \centering
    \centering
    \includegraphics[width=.8\linewidth]{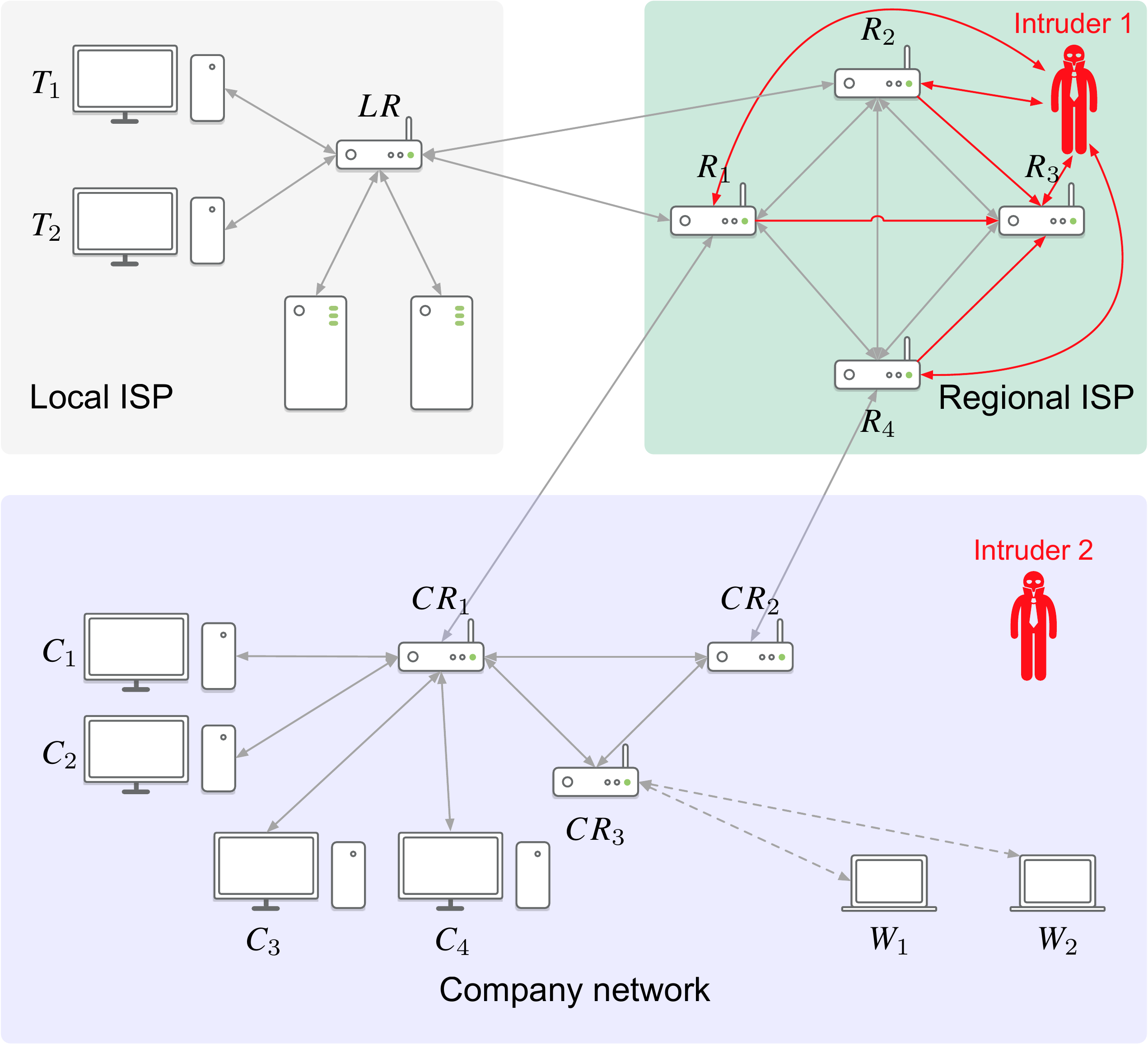}
    \caption{The sabotage game  with an attack timing pattern. }
    \label{fig:sabotage_game_timing}
\end{figure}

\paragraph{Intruders 1 \& 2.} In the presence of the two intruders, the saboteurs have a winning strategy to prevent the traveler even when the traveler has perfect observability. The strategy for the first intruder is to disable $R_1$ and for the second is to disable $CR_2$.

\paragraph{Incorporating attack timing patterns.}  We present a modification of the sabotage game presented previously~(\cref{fig:sabotage_game_timing}). In the new game presented, intruder~1, after affecting a node, needs two time steps to affect nodes $R_1$, $R_2$, and $R_4$ but needs only a time step to affect $R_3$. We accomplish this effect by modifying the edge set $E_B$ of the saboteur.  From the vertices corresponding to the routers $R_1$, $R_2$, and $R_4$, there are no out edges to each other. So, to mark these vertices, the saboteur (intruder~1) must move back to the initial vertex. However, to move and mark $R_3$, the attacker can directly move from $R_1$, $R_2$, and $R_4$ to $R_3$.

Obtaining the game graph for the sabotage game from the network diagram is straightforward, and from a modeling perspective, sabotage games are apt for dynamic network threat assessment. The complexity results tell us that the game is easy to solve from the perspective of the saboteurs. We can analyze the strategies of the saboteurs to understand the points of failure in a particular network. The stakeholders (for example, the owners, designers, cybersecurity experts, etc.) can solve problems of this type to gain insights into possible coordinated network attacks, which can then be used to erect appropriate defenses, thereby increasing network robustness.

\section{Solving dynamic network routing using~sabotage games}
\label{sec:limitBudget}
This section establishes the equivalence between sabotage games and reachability games, enabling the use of existing solvers to find attacker-defender strategies.

\subsection{Reachability games and complexity classes}
\label{sec:prelims}

This section provides a concise overview of reachability games on graphs and the various complexity classes we use.

\subsubsection{Reachability games}
A two-player \emph{reachability game}~\cite{McNaughton1993Dec,Chatterjee2011Sep,Gradel2002} is played on a directed graph $G = (V,V_R,V_S,E,F)$, which we call the \emph{arena}, between a reachability player and a safety player. The vertices $V$ of the arena are partitioned into those belonging to the reachability player $V_R$ and those belonging to the safety player $V_S$. A token is placed on a designated \emph{initial vertex} $v_{\eta}$ in $V$. The two players take turns moving the token along the graph's edges. The objective of the reachability player is to move the token to a vertex in $F \subseteq V$. The safety player tries to prevent the token from reaching a vertex in $F$. A play $\rho$ is a (possibly infinite) sequence $v_0 v_1\cdots$ of vertices such that $(v_i, v_{i+1}) \in E$, for all $0 \leq i$ and $v_0 = v_\eta$. The play $\rho$ is \emph{winning} for the reachability player if, for some $v_i$ in $\rho$, $v_i \in F$. A \emph{memoryless strategy} $\sigma\colon V_R \to V_S$ for the reachability player maps every vertex in $V_R$ to one of its successors in $V_S$. Memoryless strategies for the safety player are defined similarly. The play $\rho$ is said to \emph{agree} with the strategy $\sigma$ for player~$P$ if, $v_{i+1} = \sigma(v_i)$ whenever player $P$ has to play from $v_i$. A memoryless strategy is said to be \emph{winning} for player $P$ if all plays that agree with it win for player $P$. Such games are \emph{determined}, that is, one of the players has a memoryless winning strategy~\cite{zielonka:1998}, and the strategy for the winning player can be found in $\mathcal{O}\left(|V|+|E|\right)$ \cite{Emerson1991Oct}.~%

\subsubsection{Alternating Turing machine}
\label{sec:atms}

We present below one of the many equivalent ways of defining an ATM \cite{Chandra1981Jan} and explain the complexity classes we use. We assume that the tape alphabet $T$ is $\{\top,\bot\}$, the work tape encodes the input, and that the universal and existential states alternate. The blank tape symbol $\square$ is encoded by a fixed string $a \in T^*$. %

\begin{definition}
An ATM is a tuple $\langle Q,T,\delta,q_0,q_\mathsf{accept},q_\mathsf{reject} \rangle$ composed of the following data.
\begin{itemize}[leftmargin=*]
    \item A partition of the states into existential and universal, $Q = Q_\exists \uplus Q_\forall$.
    \item The tape alphabet, $T = \{\bot, \top \}$.
    \item A transition function $\delta = \delta_\exists \uplus \delta_\forall$, where 
    \begin{align*}
     \delta_\exists\colon& Q_\exists \times T \to 2^{Q \times T \times \{L,R\}}\text{ and}\\ \delta_\forall\colon& Q_\forall \times T \to 2^{Q \times T \times \{L,R\}}.
     \end{align*}
    \item The initial state, $q_0$.
    \item A unique accepting state, $q_\mathsf{accept}$.
    \item A unique rejecting state, $q_\mathsf{reject}$.
\end{itemize}
\end{definition}
 If an ATM is in the state $q \in Q$,  the value at the current tape head is $v$, and $\delta(q,v) = (q',v',z)$, where $z \in \{L,R\}$, then at the next step the value at the head position of the tape is $v'$, the machine will be in state $q'$, and the tape head will move left or right as given by $z$.
A \emph{configuration} of an ATM, consisting of the state $q$ and the tape contents, is \emph{final} if the state $q$ is either $q_\textsf{accept}$ or $q_\textsf{reject}$. Accepting an input string by an ATM is defined inductively, starting from final configurations that are \emph{accepting}. A final configuration is accepting if the state component is $q_\textsf{accept}$. Non-final configurations accept if the state is universal $Q_\forall$ and all the successor configurations accept or if the state is existential $Q_\exists$. At least one of the successor configurations is accepting. Finally, an ATM accepts a given input string $w \in T^*$ if the initial configuration with initial state $q_0$ and the tape's input string is accepted. 
A non-deterministic Turing machine (NDTM) is an ATM without universal states.

A deterministic Turing machine is an NDTM with $|\delta(q,t)|= 1$, for all $q \in Q$ and $t \in T$. PSPACE is the class of decision problems solvable by a deterministic Turing machine that uses several tape cells upper bound by a polynomial on the input length $n$. Formally, $$\text{PSPACE}=\bigcup_{k\geq0} \text{DSPACE}(n^k).$$ QBF is a Boolean satisfiability problem generalization in which existential and universal quantifiers can be applied to each variable. Put another way, it asks whether a quantified formula over a set of Boolean variables is true or false. QBF is the canonical complete problem for PSPACE \cite{Arora2009Apr}. 

\subsection{Traveler's perspective}
Initially, we study sabotage games in which the saboteur has a limited budget for marking (deleting) the vertices, and the traveler has limited observation.
The traveler can detect a marked vertex $v$ from a current vertex $u$ if $v \in Obs(u)$. Any vertex that the saboteur has visited and that never fell into the observation zone of the traveler so far in the game is a \emph{suspect} vertex for the traveler; that is, the traveler must avoid these vertices. Any encoding for the limited observation scenario must carry this information in the state space.

A straightforward encoding of the sabotage game in which the vertices capture \begin{itemize}[leftmargin=*]
    \item the current positions of the traveler and the intruder,
    \item the set of suspect locations,
    \item the set of marked locations (not visible to the traveler), and 
    \item the set of discovered locations.
\end{itemize}
In the above encoding, the saboteur updates the third component, which is not visible to the traveler. 
Hence, this encoding corresponds to an imperfect information game. Solving an imperfect information reachability game is an EXPTIME-complete problem~\cite{Reif1984Oct,chatterjee:2007}. 
Consequently, this naive approach would result in a 2EXPTIME cost because there is an exponential number of vertices.  %

However, this is not the true complexity of the problem.
We show that if the traveler has a winning strategy, it can be found in EXPTIME by reducing the sabotage game $G$ to a \emph{perfect information} reachability game on an arena $\bar{G}$ with an exponential (on the size of $G$) number of states. A key difference in this new encoding is that we do not explicitly maintain the set of vertices marked by the saboteur. We only intend to use this reduction to solve the sabotage game from the traveler's perspective.
Since the new reachability game itself is solvable in linear time~\cite{Mostowski1991,Emerson1991Oct}, we can find a winning strategy for the sabotage game, if one exists, in EXPTIME.  The following details are encoded into the components of the vertex set via the following data.
\begin{itemize}[leftmargin=*]
    \item The current positions of the traveler and the intruder.
    \item The set of suspect locations (locations visited by the intruder).
    \item The set of marked locations discovered by the traveler (at most $\leq m$ positions).
    \item The current turn (a record of the player to move).
\end{itemize}

\begin{theorem}
If the traveler has a winning strategy for a sabotage game~$G$, then it can be found in $\mathrm{EXPTIME}\left(\text{exp}(m,n)\right)$.
\label{theorem:inEXPTIME}
\end{theorem}
\begin{proof}
We prove the theorem for a reachability objective. 
We reduce the sabotage game played on graph $G$ to a full-information reachability game on the graph $\bar{G} = \left(\bar{V},\bar{E}\right)$ defined below. 
Let $V' = V\setminus\{v_f\}$ and $B =  2^{V'}$. Without loss of generality, we assume that the saboteur cannot move to the final vertex. %
The set of positions, $\bar{V} =  V \times V' \times B \times B \times \{t,b\}$ and $V_T = V \times V' \times B \times B \times \{t\}$. 
The initial position $v_\iota$ is~$\left(\iota_T,\iota_B,\{\iota_B\},\emptyset,t\right)$. The traveler wins if the token reaches a position in  $$\tilde{F} = \left\{(v_f,v,S,T,t) \mid v_f \in F, v \in V, S \subseteq B \text{ and } T \subseteq B \right\}.$$  
The edge set $\bar{E}$ is defined as follows.
\begin{enumerate}[leftmargin=*]
    \item \emph{The traveler moves to an unmarked, that is, non-deleted vertex}. If $u,u' \notin T$ and $(u,u') \in E_T$, then  $$(u,v,S,T,t) \rightarrow(u',v,S,T,b)$$ is an edge.
    \item \emph{The saboteur moves to an adjacent vertex and reveals the set $P \subseteq Obs(u)$ of marked vertices.} %
    $$(u,v,S,T,b) \rightarrow \bigg(u,v', \left(S\cup\{v,v'\} \right)\setminus (Obs(u)\setminus P), P \cup T,t \bigg)$$ is an edge if the following conditions are satisfied.
    \begin{enumerate}
        \item $v' \neq v_f$ and $(v,v') \in E_B$ (saboteur can only move along their edges),
        \item $|T \cup P| \leq m$  (only $m$ vertices can be marked) and
        \item $P \subseteq Obs(u) \cap S$ (set of newly detected marked vertices). 
    \end{enumerate}
        \item \emph{The traveler loops in the same vertex on losing.} If $u \neq v_f$ and $u \in T$, then 
        \begin{align*}
        (u,v,S,T,b) &\rightarrow (u,v,S,T,t)\text{ and}\\
        (u,v,S,T,t) &\rightarrow \left(u,v,S,T,b\right)
        \end{align*}
        are edges.
    \end{enumerate}
    In the edge set above, %
    the saboteur is forced to reveal the marked vertices that are observable from $u$. This edge affects the set $T$ of marked vertices known to the traveler and the set $S$ of suspect vertices. 
    
    By constructing $\bar{G}$, the traveler's winning strategy for the game $G$ exists if it can win the reachability game on $\bar{G}$. The size of graph $\bar{G} = \mathcal{O}\left(2^{n}n^{m+2}\right)$ and $E_m = \mathcal{O}\left(|\bar{G}|^2\right)$, where $n = |\bar{V}|$. If the traveler can win the reachability game, it has a memoryless winning strategy \cite{McNaughton1993Dec}. This memoryless winning strategy can be determined in $\mathcal{O}\left(|\bar{V_m} + \bar{E_m}|\right)$~\cite{Mostowski1991,Emerson1991Oct}.
Therefore, a winning strategy for the traveler can be found in time~$\mathcal{O}\left(4^n n^{2m}\right)$. 
\end{proof}

The reduction presented in the above theorem also holds verbatim for B\"uchi objectives. But, solving B\"uchi games entails a $\mathcal{O}\left(|V||E|\right)$ cost instead of the underlying $\mathcal{O}(|V|+|E|)$ cost.

We prove a matching exponential-time lower bound for the problem of determining the existence of a winning strategy for the system by reducing the problem of determining acceptance of a string by an ATM to a generalized sabotage game. 
Let $M = \langle Q,T,\delta, q_0,q_\mathsf{reject},q_\mathsf{accept}\rangle$ be an ATM that on any input of size~$n$ uses space at most~$P(n)$, where $P\colon \mathbb{N} \to \mathbb{N}$ is some fixed polynomial, that is, $M \in \mathrm{APSPACE} = \mathrm{EXPTIME}$~\cite{Chandra1981Jan}.

We construct a generalized sabotage game $G_{M,w}$ in which the traveler (system) has a winning strategy to \emph{reach} a final vertex if and only if the string $w \in \left\{\top,\bot\right\}^*$ is accepted by the ATM $M$.%
In the game $G_{M,w}$, the budget $m$ for the saboteur is one.  To prove hardness, it is enough to show the hardness concerning some observation function. Specifically, we present the reduction for the case when \emph{the traveler can view a marked vertex if and only if the traveler is adjacent to the vertex in question}. The saboteur simulates the $\forall$-player, and the traveler simulates the $\exists$-player. 
\noindent The sabotage game is played in the following three stages:
\begin{enumerate}[leftmargin=*]
    \item Initially, the players take turns to ask together and reveal the input $w$ encoded in the tape one position at a time.
    \item Then, they repeat ATM transitions while updating the tape position until an accept or reject state is reached.
    \item Finally, produce accept or reject.
\end{enumerate}

\paragraph{Stage 1 [Reading the input].} The traveler and the saboteur simulate the reading of the input using the \emph{input gadget}; we denote it by $I(w)$ and present the gadget for the input string $w = \top \cdot \bot \cdot \bot \cdot \top$ (\cref{fig:Input_gadget}). 
The vertices marked in blue are called \emph{tape vertices}; each tape cell can have the value $\top$ or $\bot$. %
The vertices colored in blue with the labels of form $P,V$ represent facts such as the tape at position $P$, which has the value $V$. We call these vertices the \emph{tape vertices}.
For example, the blue vertex~$1,\top$ corresponds to the value at the tape position~1 is $\top$. The traveler uses the vertices with the label of form~$P$ to request the saboteur to reveal the value at tape position $P$. In the example in the figure, the vertices $1,\top$; $2,\bot$; $3,\bot$ and $4,\top$ are called the \emph{vertices corresponding to the current configuration}, that is, the \emph{input configuration}.

\begin{remark}
As the saboteur and traveler simulate the ATM, the current tape configuration vertices change accordingly.
\end{remark}

The two players are initially located in the vertices marked as saboteur start and traveler start. The traveler begins the game by moving to the white vertex~1. This action corresponds to the traveler requesting the saboteur to reveal the value of the tape at position~1.
In the example, the value of the tape at position~1 is~$\top$. The saboteur has one of the following choices:

\begin{enumerate}[label=(\alph*),leftmargin=*]
    \item Incorrectly moves to $1,\bot$: In this case, the traveler can move to the orange vertex labeled \emph{input escape} and win the game in the next round.
    \item Correctly reveals the tape value by moving to the vertex $1,\top$. In this case, the traveler must not move to the input escape vertex (orange) in the next round, as she (the traveler) would subsequently lose the game. Indeed, the saboteur can catch the traveler in the next step using the dotted red edge.
\end{enumerate}

\begin{figure}[!t]
    \centering
    \includegraphics[width=.75\linewidth]{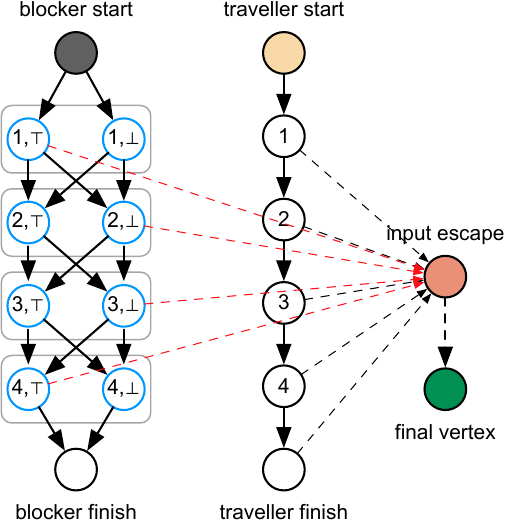}
    \caption{Input gadget $I(w)$ for the input $w = \top \cdot \bot \cdot \bot\cdot \top$. %
    }
    \label{fig:Input_gadget}
\end{figure}

For the input $w = \top \cdot \bot \cdot \bot \cdot \top$ (\cref{fig:Input_gadget}), the saboteur moves through the vertices $1,\top \rightarrow 2,\bot \rightarrow 3,\bot \rightarrow 4,\top$ and finally move to the vertex saboteur finish. Since the distance of the tape vertices (blue) from the white vertices is greater than one, the saboteur cannot observe these vertices. %
The following proposition is now direct.

\begin{proposition}
\label{prop:inp_gad}
In the input gadget $I(w)$, when the saboteur and the traveler reach vertices, saboteur finish and traveler finish (respectively), the traveler suspects the vertices corresponding to the input tape configuration.
\end{proposition}
In other words, after $|w|+1$ rounds, the traveler suspects the vertices corresponding to the input configuration of the tape.  We show that the following invariant is maintained during the simulation of the ATM in the game.

\begin{tcolorbox}[breakable, pad at break*=1mm,
    colback=gray!5, arc=0pt, outer arc=0pt, boxrule=0pt, frame hidden]
\noindent \textbf{Invariant:} After the simulation of every ATM action, the traveler suspects \emph{vertices corresponding to the current configuration} of the tape. Moreover, the traveler knows that the other tape vertices are not marked.
\end{tcolorbox}

\paragraph{Stage 2 [Simulating ATM's, $M$, transitions on input, $w$].} 
Say the current state of the ATM is $q \in Q_\forall$ and suppose the head is located at a position $1 \leq p \leq |w|$. In the ATM, the $\forall$-player inspects the value at the tape position $p$ and decides the next state of the ATM while simultaneously the tape position $p$ with a new value before moving the tape head left or right. 
In the game $G_{M,w}$, the saboteur and the traveler simulate this transition with two sub-stages.
\begin{enumerate}
    \item[S1)] First, the two players agree on the current value at the tape position $t$ and then erase this tape position; that is, the traveler ensures that the vertices $p,\top$, and $p,\bot$ are not marked by the saboteur.
    \item[S2)] Next, the two players simulate the transition by moving to their vertices corresponding to the correct state and tape position. On the way to this vertex, the saboteur moves through the vertex $p,val$ corresponding to the new value $val$ of the tape position $p$ after the transition. Consequently, the traveler once again suspects the vertices corresponding to the current configuration of the tape.
\end{enumerate}

\begin{figure*}[t!]
    \centering
    \includegraphics[width=.8\linewidth]{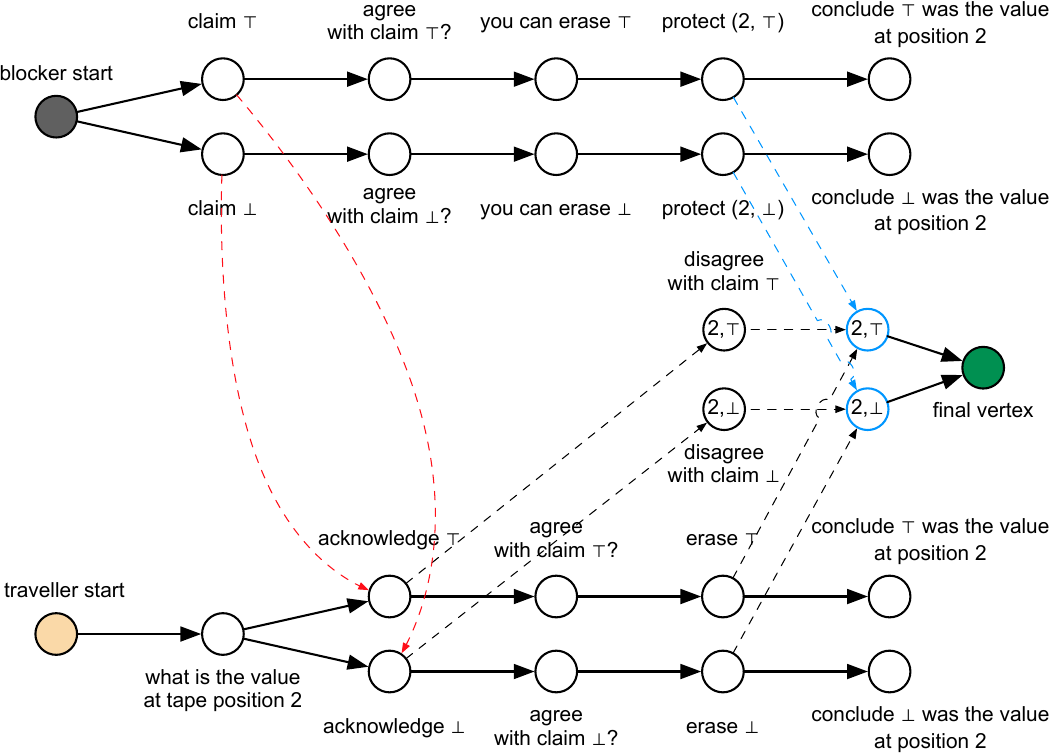}
    \caption{When the current state of the ATM  is $q$, the depicted tape eraser gadget $E(q,2)$ enables the traveler to check whether the vertex corresponding to the current value of the tape position~2 is marked.}
    \label{fig:eraser_gadget}
\end{figure*}

The eraser gadget erases the value at tape position 2 when the current state of the ATM is $q$ and is denoted $E(q,2)$ (\cref{fig:eraser_gadget}). 
The final vertex and the vertices $2,\top$ and $2,\bot$ are common to the input gadget $I(w)$. Precisely, the gadget does the following.

\begin{enumerate}[leftmargin=*]
    \item In the gadget $E(q,2)$, the traveler and saboteur are initially located in the marked start vertices.
    \item  The traveler asks the saboteur to reveal the current value at the tape position~2. The saboteur has one of two choices: claim the value is $\top$ or claim $\bot$ by moving to either the vertex claim~$\top$ or claim~$\bot$ (resp.).  The traveler has to correctly acknowledge the claim by moving appropriately to the vertex with the label acknowledge~$\top$ or acknowledge~$\bot$ (resp.). If not, the saboteur can catch the traveler in the next step by using the red dashed edge. 
    \item Suppose the current value of the tape position 2 is $\top$, then the traveler suspects the blue vertex $2,\top$ (see invariant). One of the following scenarios must have occurred in the previous step.
    \begin{enumerate}
        \item \emph{The saboteur correctly claims $\top$}: Since the traveler suspects the blue vertex $2,\top$, the traveler must not move to the vertex labeled disagree with claim $\top$.
        \item \emph{The saboteur incorrectly claims $\bot$}: The traveler can win the game by moving to the vertex that disagrees with the claim $\bot$ and eventually to the final vertex.
    \end{enumerate}
    \item  Next, the saboteur and traveler do not have a choice for the subsequent vertices they must move to. Continuing with the case where the value of the tape position~2 is $\top$, the saboteur and traveler move to the vertices with labels \emph{you can erase~$\top$} and \emph{erase the value~$\top$} (resp.). Now, the traveler can check whether the saboteur has placed a trap at $2,\top$. However, the traveler cannot move to the tape vertex $2,\top$ in the next step as the saboteur would catch the traveler in the next move.
\end{enumerate}
\begin{figure*}[!t]
    \centering
    \subfloat[$\forall$-gadget\label{fig:forall}]{
        \includegraphics[width=.8\linewidth]{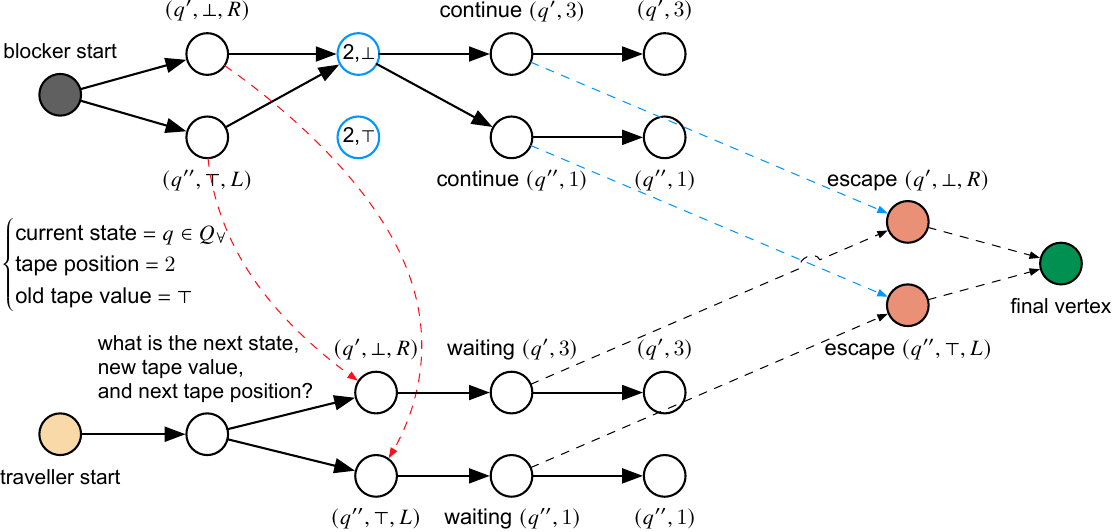}
    }    \\
    \subfloat[$\exists$-gadget\label{fig:exists_gad}]{
        \includegraphics[width=.8\linewidth]{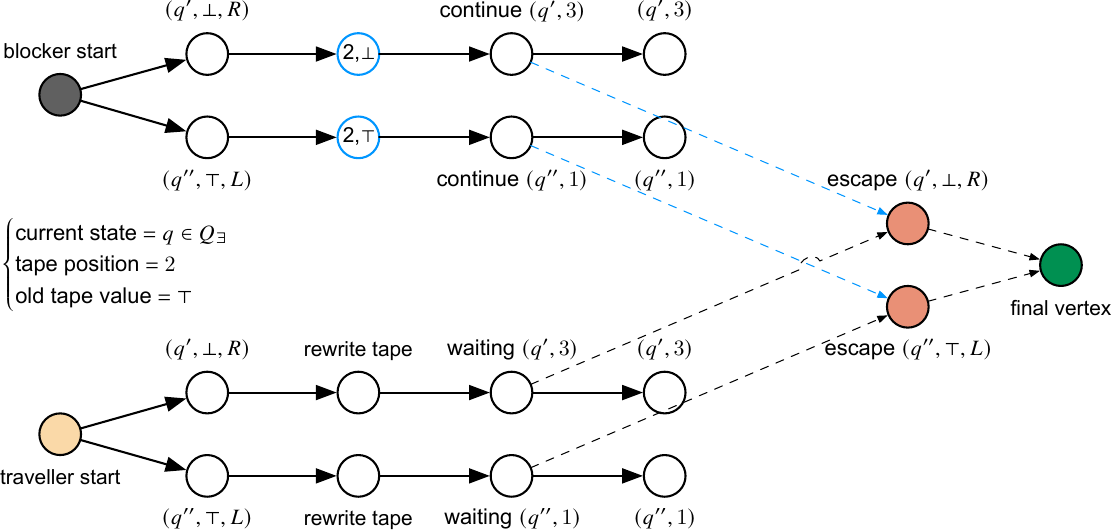}
    }
    \caption{(a) The gadget $\forall(q,2,\top)$ models the transition when the current state is $q$, current tape position is $2$, and the value of the current tape position is $\top$. (b) The gadget $\exists(q,2,\top)$ models the transition when the current state is $q$, the current tape position is $2$, and the value of the current tape position is $\top$.}
\end{figure*}

\begin{proposition}
\label{prop:eraser_gad}
Suppose the current tape position is $p$; when the traveler and the saboteur reach the conclude vertices in the eraser gadget $E(q,p)$, then the traveler knows whether the saboteur has marked the tape vertices $p,\top$ and $p,\bot$.
\end{proposition}
Following the erasure of the head of the tape, the actual transition occurs. 
We accomplish the simulation of this ATM action in the game with the help of two different types of gadgets. 
The type of gadget used depends on whether the state of the ATM is a $\forall$ state or a $\exists$ state. The subtlety is that we have a separate gadget for every transition for each tape position.
We present the $\forall$-gadget for the transition $\delta_\forall (q,\forall) = \left\{(q',\bot,R),(q''\bot,L)\right\}$ when the current tape position is $2$ (\cref{fig:forall}). 
\begin{enumerate}[leftmargin=*]
    \item The traveler and saboteur start in the marked start vertices. Following this, the traveler asks the saboteur to choose the next state, the next tape position, and the value for the current tape position. 
    \item The saboteur uses the red edges to ensure the traveler acknowledges this choice correctly by moving only to the corresponding vertex.
    \item The saboteur moves to the vertex (blue) corresponding to the correct tape value (for the current tape position). In the gadget presented in the figure, the saboteur moves to the blue vertex with the label $2,\bot$.
    \item The traveler moves to the corresponding wait vertex, that is, either the vertex labeled \emph{waiting}~$(q',3)$ or \emph{waiting}~$(q'',1)$. Following this move, the saboteur must move to a continuous vertex in the same gadget to prevent the traveler from using the escape vertices.
\end{enumerate}

\begin{remark}
The orange escape vertices are unique to the respective gadgets.
\end{remark}
The saboteur has to return to the vertices above in this gadget from the tape vertices to ensure that the traveler cannot use the escape vertex (in this gadget) to move to the final vertex.

\begin{proposition}
\label{prop:forall_gad}
Suppose the current tape position, state, the tape value is $p$, $q \in Q_\forall$, and $\top$ (resp.). If the traveler and the saboteur simulate the transition $\delta_\forall(q,\top)$ using the gadget $\forall(q,p,\top)$, then the traveler suspects the vertices corresponding to the current tape configuration. 
\end{proposition}

We present the $\exists$-gadget for transition $ \delta_\exists (q,\forall) = \left\{(q',\bot,R),(q''\top,L)\right\}$ when the current position is $2$ (\cref{fig:exists_gad}). %
\begin{enumerate}[leftmargin=*]
    \item The traveler begins by choosing the next state, the new tape value, and the tape position. The traveler has two choices in the gadget presented, which are determined by the transition function.
    \item Following the traveler's move, the saboteur can either follow the traveler and make the same or the opposite choice, meaning that depending on the structure of the transition function, in the example, the saboteur only has two choices. 
    \item The following move of the traveler corresponds to waiting for the saboteur to set the new value of the tape before moving the tape head to the new position. 
    \item To prevent the traveler from using an escape vertex to reach the final vertex, the saboteur in the previous step must have moved to the correct tape vertex (in turn, make the same choice for the transition as the traveler). 
\end{enumerate}

\begin{proposition}
\label{prop:exists_gad}
Suppose the current tape position, state, the tape value is $p$, $q \in Q_\exists$, and $\top$ (resp.). If the traveler and the saboteur simulate the transition $\delta_\exists(q,\top)$ using the gadget $\exists(q,p,\top)$, then the traveler suspects the vertices corresponding to the current tape configuration. 
\end{proposition}

\paragraph{Stage 3 [Simulate the acceptance or rejection of a string~$w$].}
The game $G_{M,w}$ does not have the eraser, $\forall$, or $\exists$ gadgets for the accept and reject states. Because once the ATM reaches one of these states, the string is either accepted or rejected.
We present the \emph{accept gadget} when the current state and tape position is $q_\mathsf{accept}$ and $2$ (\cref{fig:accept_reject}~(a)).

Initially, the traveler and saboteur begin in marked vertices. The traveler reaches the final (accepting) state and wins the game in one move. 
Consequently, the string $w$ is accepted.
We present the \emph{reject gadget} for the case when the current tape position is $2$ (\cref{fig:accept_reject}~(b)). In the gadget, when the traveler reaches the reject state, there is no way to come out of this state. In the next move, the saboteur wins the game and consequently rejects the string $w.$

We now play the entire game $G_{\mathcal{M},w}$ for a specific ATM of the following form $\mathcal{M} = \langle Q,T,\delta,q_0,q_\mathsf{reject},q_\mathsf{accept}\rangle$, where 
\begin{itemize}[leftmargin=*]
    \item $Q = \left\{q_0,q_1,q_\mathsf{accept},q_\mathsf{reject}\right\}$, $Q_{\exists} = \left\{q_0,q_\mathsf{accept},q_\mathsf{reject}\right\}$ and $Q_{\forall} = \left\{q_1\right\}.$
    \item $T = \{\top,\bot\}.$
    \item The transition function $\delta$ is defined as follows
    \begin{itemize}
        \item $\delta(q_0,\top) = \left\{(q_0,\top,R)
        \right\}$,
        \item $\delta(q_0,\bot) = \left\{(q_\mathsf{accept},\bot,R)\right\}$,
        \item $\delta(q_1,\top) = \left\{(q_0,\top,R),(q_1,\bot,L) 
        \right\}$, and
        \item $\delta(q_1,\bot) = \left\{(q_1,\bot,L),(q_\mathsf{reject},\top,R)\right\}$
    \end{itemize}
\end{itemize}
 for the specific input $w = \top \cdot \bot \cdot \bot$ (\cref{fig:overall}).

The players first read the input by moving along the black edges of the input gadget. The players then simulate the ATM in position one, and the current value of this position is $\top$. The players together erase this value at the tape position one and simulate the $\delta_\exists$ transition ($q_0 \in Q_\exists$) and write the new value for the tape using the orange edges. There is only one choice for the traveler ($\exists$-player), the current state is now $q_1$, and the tape position is two. Again, the players erase the value at position two and simulate the $\forall$-transition using the gadget $\forall(q_1,2,\bot)$ through the violet edges. Lastly, the string is directly rejected by moving through the blue link to the reject gadget.

Regarding the simulation of an arbitrary ATM $M$ for a given input $w$,  even though the saboteur has a budget of one $(m=1)$, she can never mark (delete) a vertex. If the saboteur marks a vertex, then the traveler can find the marked vertex using the dashed edges starting from \emph{erase the value} vertices in the eraser gadgets (\cref{fig:eraser_gadget}). The next time the saboteur and the traveler visit a different eraser gadget, the traveler can use the \emph{disagree with claim} vertices to reach the final vertex as the traveler knows that the saboteur has no budget left.

\begin{proof}[Proof of the invariant]
After the traveler and saboteur finish simulating an ATM action, the following three scenarios must have occurred before the previous move.
\begin{enumerate}[leftmargin=*]
    \item The saboteur and the traveler played in the input gadget $I(w)$: In this scenario, by Proposition~\ref{prop:inp_gad}, the traveler suspects the vertices corresponding to the input.
    \item The saboteur and the traveler played in a $\exists$ gadget: In this scenario (\cref{fig:exists_gad}), the onus is on the traveler to first choose the $\delta_\exists$ transition. If the saboteur chooses a different choice than the traveler's, the traveler can use an orange vertex to reach the final vertex. Thus, the saboteur is forced to make the same choice. For the same reason, after the saboteur moves to a new tape value vertex, the saboteur must come back to the current $\exists$ gadget.
    
    \item The saboteur and the traveler played in a $\forall$ gadget: In this scenario (\cref{fig:forall}), the onus is on the saboteur to first choose the $\delta_\forall$ transition. If the traveler chooses a different choice other than the choice made by the saboteur, the saboteur can capture the traveler in the following move using a red-colored edge.  In the gadget, when the saboteur visits a vertex with a new value on the tape, the traveler can use the orange vertices to escape to the final vertex, ensuring that the saboteur comes back to the same gadget. %
\end{enumerate}

\begin{figure}[!t]
    \centering
    \subfloat[The accept gadget $A(q_\mathsf{accept}, 2)$ when the current state is $q_\mathsf{accept}$ and tape position is 2.\label{subfig:accept}]{
        \includegraphics[width=.336\textwidth]{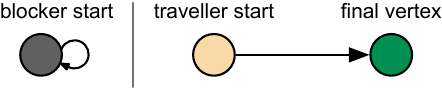}
    }\hfill %
    \subfloat[The reject gadget $R(q_\mathsf{reject},2)$ when the current state is $q_\mathsf{reject}$ and tape position is~2.\label{subfig:reject}]{
        \includegraphics[width=.264\textwidth]{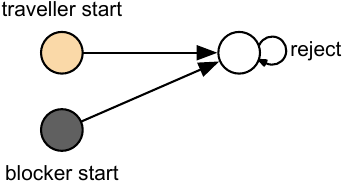}
    }
    \caption{The final step requires accept and reject gadgets.}
    \label{fig:accept_reject}
\end{figure}

\begin{figure*}[!t]
    \centering
     \includegraphics[width=.8\linewidth]{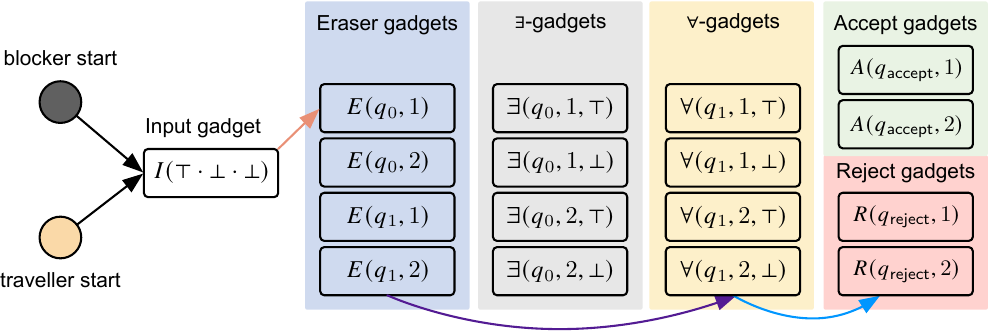}
    \caption{The game $G_{\mathcal{M},w}$ for the ATM $\mathcal{M}$ and input string $w = \top \cdot \bot \cdot \bot$. We show the sequence of gadgets in the simulation of this string but only link the sequence of gadgets used in the simulation.   The saboteur rejects the string by using the blue edge to move to the reject gadget $R(q_\mathsf{reject},2)$ instead of moving to $E(q_1,1)$.}
    \label{fig:overall}
\end{figure*}

We induct the sequence of the gadgets used in the ATM. In the base case, the saboteur and the traveler play in the input gadget $I(w)$. We have already proved this case. Let us assume inductively that the invariant is true for the first $p \in \mathbb{N}$ transitions. Let $q^p, H^p$, and $v^p$ be the current state, the current tape position, and the value of the tap position after $p$ transitions. Let $q^p \in Q_\exists$. 
In the sabotage game, the traveler and the saboteur have to erase the head of the current tape using the eraser gadget $E(q^p,p)$. As a result of proposition~\ref{prop:eraser_gad}, the traveler knows if the tape vertex $H^p,v^p$ is marked. However, as a consequence of scenarios 2 and 3 discussed above, the traveler suspects the vertex corresponding to the new value stored at $H^p$ after the end of the transition.
\end{proof}
The invariant, in turn, implies the following theorem.

\begin{theorem}
The traveler has a winning strategy for the sabotage game on $G_{M,w}$ if and only if the ATM $M$ accepts the input~$w$.
\end{theorem}

The following holds as a consequence of the above result, the fact that $\mathrm{APSPACE} = \mathrm{EXPTIME}$~\cite{Chandra1981Jan} and Theorem~\ref{theorem:inEXPTIME}.

\begin{corollary}
The problem of finding a winning strategy for the traveler (if one exists) is $\mathrm{EXPTIME}$-complete.
\label{thm:reach_exptime}
\end{corollary}

\subsection{Saboteurs perspective}

To solve a sabotage game $G$ from the saboteur's perspective, we assume that the budget $m$ is fixed and $Obs(v) = V$ for all $v \in V$, that is, in the worst case from the \emph{perspective of the saboteur} the traveler can correctly guess the location of marked vertices.

We encode the sabotage game into a B\"uchi game $\widetilde{G}$.
In addition to the current positions of the traveler and the saboteur, we also encode into the state space of the new arena $\widetilde{G}$ the set of marked locations. These locations correspond to the set of locations removed by the saboteur but are visible to the traveler. This reduction is the standard subset construction~\cite{doyen:2011}. %
A winning strategy for the saboteur in the reachability game on $\widetilde{G}$ is a strategy for the sabotage game. 

\begin{theorem}
If the saboteur has a winning strategy for a sabotage game, it can be found in exponential in $m$ but polynomial in $n$.
\end{theorem}
\begin{proof}
By constructing $\widetilde{G}$, the saboteur has a winning strategy if and only if she has a winning strategy for the reachability game on $\widetilde{G}$. The size of graph $\widetilde{G}$ is $\mathcal{O}(n^m)$, and the number of edges is square to the number of vertices (size). Lastly, reachability games can be solved in time $\mathcal{O}(|\text{vertices}|+|\text{edges}|)$.
\end{proof}

\begin{corollary}
If the budget for the saboteur $m$ is fixed, then a winning strategy can be found, if it exists, for the saboteur in time polynomial in the size of the sabotage game. 
\end{corollary}

\section{Exploiting a bound on the number of rounds to~use~QBF solvers}
We specifically analyze the sabotage game under the scenarios when
(i) the saboteur has an unlimited budget (\textbf{B1}) or (ii)
the saboteur has a fixed (limited) budget (\textbf{B2}) and the traveler cannot observe the deletions (\textbf{T2}).

The idea is to polynomially bind the number of rounds the traveler has to reach the final vertex. More specifically,  if the traveler can reach the last vertex $v_f$, she (traveler) can reach $v_f$ in at most $n^{3}$ rounds, where $n = |V|$. We derive this bound for the number of rounds with the help of the two-player reachability game reduction presented earlier. %

We encode the sabotage game into a QBF formula of polynomial size, showing its containment in PSPACE \cite{Alur2005Apr}.

\begin{proposition}
In the game $\bar{G}$, if the traveler can win the reachability game, then it can reach a final vertex in $F$ in at most $\mathcal{O}(n^{3})$ rounds.
\label{prop:n3bound}
\end{proposition}
\begin{proof}
If the traveler has a winning strategy for the reachability game, then it has a memoryless winning strategy \cite{Mostowski1991,Emerson1991Oct}. 
In any valid play that respects a memoryless winning strategy, the second component that stores the set of vertices visited by the saboteurs cannot decrease. This component can remain unchanged for at most $n^{2}+1$ rounds, or the players will hit a losing loop for the traveler, that is, a loop that the saboteurs can enforce. The size of the second component is at most $n-1$ (it can be $V-\{v_f\}$). Hence, the traveler has to reach a final vertex in at most $(n-1)(n^{2+1}+1)$ rounds.
\end{proof}

\begin{remark}
If the traveler plays against $k$ saboteurs and has a winning strategy, then the traveler can reach the final vertex in at most $\mathcal{O}\left(n^{k+2}\right)$ rounds.
\end{remark}

\begin{proposition}
The problem of determining whether the traveler has a winning strategy in the scenario above is contained in $\mathrm{PSPACE}$.
\label{prop:containtedPSPACE}
\end{proposition}
\begin{proof}
If the traveler reaches $v_f$, then exists a strategy using which the traveler can reach $v_f$ in at most $\gamma = n^3$ rounds.%
Let $\rho = \left(t_0,b_0\right),\cdots, \left(t_i,b_i\right),\cdots, \left(t_\gamma,v_\gamma\right)$ be a play, where $t_0 = \iota_T$ and $c_0 =\iota_B$. Let $e\colon V \times V \to \{\top,\bot\}$ be a predicate.
\[
e(u,v) = \begin{cases}
\top, \text{ if } (u,v) \in E, \\
\bot, \text{ otherwise }.
\end{cases}
\]

We encode the sabotage game in a QBF formula $\psi$ with $\gamma$ existential variables, $\gamma-1$ universal variables, and $2(\gamma-1)$ quantifier alternations recursively.
{
\begin{align*}
    &\psi \triangleq  \exists r_1\colon e\left(r_0,r_1\right) \land \forall c_1 \neq v_f\colon \\   & \left( e\left(c_0,c_1\right) \land r_1 \neq v_f \right) \Rightarrow \left(r_1 \neq c_0 \land r_1 \neq c_1 \land \psi_1 \right). \\
    &\psi_1 \triangleq \exists r_2\colon e(r_1,r_2) \land \forall c_2 \neq v_f\colon \\  & \left(e\left(c_1,c_2\right) \land r_2 \neq v_f \right) \Rightarrow  \left(r_2 \neq c_0 \land r_2 \neq c_1  \land r_2 \neq c_2 \land  \psi_2 \right). \\
     & 
     \qquad \qquad \qquad \quad \qquad \vdots \\
    &\psi_{i-1} \triangleq \exists t_{i}\colon e\left(t_{i-1},r_i\right) \land \forall b_{i} \neq v_f\colon \\  &
    \left(  e\left(b_{i-1},c_i\right) \land  r_i \neq v_f \right) \Rightarrow \underbrace{\bigwedge_{0\leq j \leq i} (r_i \neq c_j)}_{\theta_{i-1}} \land \psi_{i}. \\
     & \qquad \qquad \qquad \quad  \qquad \vdots \\
    &\psi_{\gamma-1} \triangleq \exists t_\gamma\colon e(t_{\gamma-1},t_\gamma) \land t_\gamma = v_f.
\end{align*}
}
The formula $\psi_i$ ensures that the traveler and the saboteur move along the edges by checking if $\left(r_i,t_{i+1}\right)\in E$ and $\left(c_i,b_{i+1}\right)\in E$, respectively. If the traveler reaches the final vertex $v_f$, the formula $\psi$ is satisfied. If not, that is, $t_{i+1} \neq v_f$, then we check if the traveler has not visited any of the vertices previously visited by the saboteur. We check that $\theta_{i}$ is satisfied, that is, if $t_{i+1} \notin \left\{c_0,c_1,\dots,b_{i+1}\right\}$. If the traveler has not reached the final vertex and she has not been captured so far, then the sub-formula $\psi_{i+1}$ should be satisfied. Thus, if the formula $\psi$ is satisfiable, then the traveler has a winning strategy that can be recovered from a model for $\psi$.
\end{proof}

Any classic sabotage game can be encoded into an equivalent generalized sabotage game in which the saboteur has an unlimited budget. 
Say the classic sabotage game is played on the (undirected) graph $\bar{G} = \left(\bar V, \bar E\right)$ and $\bar F$ is the set of final vertices, we construct a new sabotage game $G = \gtuple$. The new sabotage game $G$ is played on the set of vertices $V = \bar{V} \cup \left\{v_e \mid e \in E\right\} \cup \{z\}$, that is, a new vertex is introduced for every edge in the original graph between the two endpoints of the old edge. Additionally, $V$ contains one more vertex $z \notin \bar{V}$. The edges for the traveler $E_T$ are defined as follows. If $e = (u,v)$ is an edge in $E$, then $\left(u,v_e\right)$ and $\left(v_e,v\right)$ are edges in $E_T$. Additionally, $\left(v_e,u\right)$ and $\left(v,v_e\right)$ are also edges in $E_T$ (these edges capture the other direction). The traveler has to perform two actions to make one move in the original graph. The edges for the saboteur $E_B = \left\{\left(v_e,z\right) \mid e \in E \right\} \cup \left\{\left(z,v_e\right) \mid e \in E \right\}$. %

In the new sabotage game, the saboteur is initially at $z$, that is, $\iota_B = z$. The construction ensures that in every alternate move, the saboteur can delete any vertex $v_e$. %
The budget $m$ for the saboteur is $\infty$.
The traveler is initially at vertex $\iota$, then $\iota_T =\iota$ until the final vertex $F = \bar{F}$. %

\begin{remark}
\label{thm:unlimited}
The traveler has a winning strategy in sabotage game~$G$ if she wins in the classic sabotage game on $\bar{G}$.
\end{remark}

\begin{corollary}
The above remark, in conjunction with Proposition~\ref{prop:containtedPSPACE} and the PSPACE-hardness of the classical sabotage game~\cite{Loding2003Aug} implies that solving the sabotage game~when the saboteur has an unlimited budget is PSPACE-complete.
\end{corollary}
\subsection{Randomized Experiments}

We show the asymmetric hardness of solving generalized sabotage games from the perspective of the attacker and the defender through experiments. We record synthesis time for finding defender and attacker strategies on sabotage games on networks generated uniformly at random with the number $n$ of locations as the parameter. In each random network instance, each location is connected to every other location for both the attacker and the defender. 
For each data point, we generate ten random instances by fixing all the other parameters and record the synthesis time as the average over these ten instances~(\cref{fig:bench}) under the assumption that the defender can a) use at most two firewalls and b) scan an unlimited number of times. We performed the experiments on macOS with an M1 processor and 8 GB memory.

\begin{figure}[t!]
\centering
\begin{tikzpicture}
\begin{axis}[
    axis lines*=left,
	xlabel={number of locations},
	ylabel={synthesis time [\si{\second}]},
	enlarge y limits=false,
    enlarge x limits=false, 
	width=.45\columnwidth, height=.45\columnwidth,
     ymax=6.5,
    xtick={10,40,80},
    ytick={0,6.5},
        tick style={black,thick},
    clip mode=individual,
    x tick style={yshift=-0.4cm},
    x tick label style={yshift=-0.4cm},
    x axis line style={yshift=-0.4cm},
    y axis line style={xshift=-0.4cm},
    y tick style={xshift=-0.4cm},
    y tick label style={xshift=-0.4cm},
    yticklabel style={
  /pgf/number format/precision=3,
  /pgf/number format/fixed},
    ]

\addplot[color=blue,mark=*,thick] coordinates {
	(10, 0.01)
	(20, 0.02)
	(30, 0.05)
	(40, 0.08)
	(50, 0.09)
	(80, 0.27)
};
\addplot[name path=meanplusvar,color=blue!60] coordinates {
	(10, 0.00)
	(20, 0.00)
	(30, 0.02)
	(40, 0.07)
	(50, 0.07)
	(80, 0.25)
};
\addplot[name path=meanminusvar,color=blue!60] coordinates {
	(10, 0.02)
	(20, 0.04)
	(30, 0.08)
	(40, 0.09)
	(50, 0.11)
	(80, 0.29)
};

\addplot [red!30] fill between[of=meanplusvar and meanminusvar];
\end{axis}
\end{tikzpicture}
\begin{tikzpicture}\begin{axis}[
    axis lines*=left,
	xlabel={number of locations},
	enlarge y limits=false,
    enlarge x limits=false, 
	width=.45\columnwidth, height=.45\columnwidth,
    xtick={10,12,13,15},
    ytick={0.04,6.5},
        tick style={black, thick},
    clip mode=individual,
    x tick style={yshift=-0.4cm},
    x tick label style={yshift=-0.4cm},
    x axis line style={yshift=-0.4cm},
    y axis line style={xshift=-0.4cm},
    y tick style={xshift=-0.4cm},
    y tick label style={xshift=-0.4cm},
    yticklabel style={
  /pgf/number format/precision=3,
  /pgf/number format/fixed},
      xticklabel style={
  /pgf/number format/precision=3,
  /pgf/number format/fixed},
    ]
\addplot[color=red,mark=*,thick] coordinates {
	(10, 0.04)
	(12, 0.15)
	(13, 1.3)
	(15, 6.5)
};
\end{axis}
\end{tikzpicture}
\caption{Strategies are easy to compute for the defender (left) and hard for the attacker (right).}
\label{fig:bench}
\end{figure}
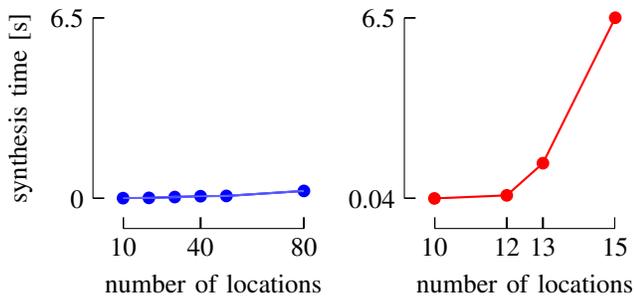

Even when the defender can use only two firewalls, the synthesis time for an attacker's strategy grows exponentially. For the experiment with $20$ locations, the resulting reachability game for the attacker has $\mathcal{O}(2^{20})$ vertices, and the experiment timed out after $10$ minutes. In contrast, we synthesized the strategies for the defender in less than a second. We find that from the defender's perspective, the number of vertices does not affect the scalability of the technique.
We analyze sabotage games with a number of locations ranging from $10$ to $100$. This range is characteristic of the number of \emph{middleboxes} in enterprises of small-medium size~\cite{sherry:2012}.

\section{Related work}
\label{sec:related_work}

Classic sabotage games~\cite{vanBenthem2005} end in a polynomial number of rounds because the saboteur can remove all edges (links) in polynomial time. Consequently, the game is at most PSPACE-complete on the size of the graph~\cite{Loding2003Aug}.  

Rohde proposes and studies the path sabotage logic to balance the asymmetry between the traveler ‘moving’ and the saboteur ‘deleting.’ Movements are local, whereas there is a global choice for edge deletion. The complexity of model checking for PSL is also PSPACE-complete~\cite{Rohde2004MovingIA}.
The decision problem for the randomized version of the sabotage game is also PSPACE-complete~\cite{Klein2012Jul}. A characterization of sabotage games with multiple destinations on weighted trees has been used to analyze dynamic network reliability or vulnerability~\cite{Kvasov2016Mar}. 

In the quantitative analysis of sabotage games, the saboteur is given a budget
that she distributes amongst the edges of the graph. At the same time, the traveler attempts to minimize the quantity of budget witnessed while completing his task. The problem of determining if Runner has a strategy to ensure a cost below some threshold is EXPTIME-complete, while if the budget is fixed apriori, then the game can be solved in PTIME~\cite{brihaye2015quantitative}. 
Even though the budget restricts the saboteur, the main point of difference is that in our setting, the traveler loses the game when she visits a marked vertex. The complexity for the traveler comes due to her incomplete knowledge.

Dynamic network games~\cite{Radmacher2008Feb,Gruner2013Jul} played between a constructor (traveler) and a destructor (saboteur) vary from generalized sabotage games in the following aspects: both players change the structure of the graph, there is perfect information, and the objective of the constructor is to maintain network connectivity. 
In poison games---yet another game on evolving structures---the saboteur can only remove vertices adjacent to the traveler's current location~\cite{Duchet1993May}.

Trap games~\cite{bonato:2011} that are a variant of cops and robber~\cite{Aigner1984Apr} in which the cops can use traps can be viewed as a generalization of the classic sabotage game. The main result, however, is from a graph structural point of view; the authors show that one cop and a fixed number of visible traps are insufficient to capture a robber on a graph where two cops without traps can catch the robber. 
In trap games, both players have perfect information, and there is no constraint on the budget, assumptions we relax in generalized sabotage games. 
Additionally, we approach the problem from an algorithmic and complexity-theoretic point of view. The computational complexity of the classic game of cops and robbers and its variants, like the guard game, etc., is EXPTIME-complete~\cite{Samal2011,Mamino2013Mar}. However, the hardness of these games comes from the variable number of cops rather than incomplete information.

Attack graphs~\cite{jha:2002,sheyner:2002} and attack trees~\cite{schneier:1999} represent the attacker's exploit paths, which take advantage of system vulnerabilities and are helpful in model-based network security verification. Attack-defense trees address the lack of defender actions in attack trees~\cite{kordy:2010}. Attack-defense trees can be encoded into stochastic two-player games to study the interaction between players and account for quantitative and probabilistic aspects of security scenarios~\cite{aslanyan:2016}. %

Two-player games are used for finding strategies for active deception~\cite{abhishek:2020}. Partially observable stochastic games capture the interaction between an attacker and a defender with one-sided partial observations~\cite{horak:2019}. In this game, the attacker aims to exploit and compromise the system without being detected with complete observation, and the defender is to detect the attacker and reconfigure the honeypots. On the contrary, in our work, it is the attacker (saboteur) that marks (disables) the nodes, while the defender (traveler) has to navigate the network while avoiding these nodes. Moreover, the traveler does not have the defender covers each targetmation games model an attacker with incomplete information and form a belief about the defender’s unit~\cite{horak:2017}. The players employ Bayesian rules to update their beliefs about the state of the game. An analysis of this scenario using Stackelberg games assumes that the attacker knows the probability that each target is covered by the defender but is oblivious to the detailed timing of the coverage schedule~\cite{vorobeychik:2012}. In contrast, in our formulation, the observability of disruptions is a function of the current location of the defender and a fixed budget restricts the attacker.

\section{Conclusion}
\label{sec:conclude}

Defenders can proactively adapt their security measures in response to evolving threats by modeling dynamic network routing with sabotage games. In particular, our ability to determine the potential for network disruption when some nodes are already compromised can aid in designing resilient security strategies that prioritize protecting critical nodes and rapidly mitigating threats. In realistic attack pattern scenarios, the complexity results of sabotage games are tractable for the defender. A practical advantage or encoding security as a reachability game is knowing whether the saboteur can disrupt the network when specific nodes are already compromised.  Stakeholders can leverage the sabotage games modeling framework to strategically enhance network security through a ``what-if'' analysis for fortifying potentially compromised nodes.

\clearpage
\IEEEtriggeratref{23}
\printbibliography

\end{document}